\documentclass[conference,a4paper]{IEEEtran}

\usepackage[utf8]{inputenc} 
\usepackage[T1]{fontenc}
\usepackage{url}
\usepackage{ifthen}
\usepackage{cite}
\usepackage[cmex10]{amsmath}

\usepackage{amsmath,amssymb,amsthm,cite,amsopn}
\usepackage{graphicx}
\usepackage{color}
\usepackage{tikz}
\usetikzlibrary{patterns}
\usepackage{comment}

\newtheorem{theorem}{Theorem}
\newtheorem{example}{Example}
\newtheorem{corollary}{Corollary}

\DeclareMathOperator{\sgn}{sgn}
\newcommand{\rmv}[1]{}
\usepackage{algorithm}
\usepackage{algorithmic}

\begin{document}
\title{Service Rate Region of Content Access from Erasure Coded Storage}

 \author{
   \IEEEauthorblockN{Sarah E. Anderson\IEEEauthorrefmark{1},
                     Ann Johnston\IEEEauthorrefmark{2},
                     Gauri~Joshi\IEEEauthorrefmark{5}, 
                     Gretchen L. Matthews\IEEEauthorrefmark{3},
                     Carolyn Mayer\IEEEauthorrefmark{4}, and Emina~Soljanin\IEEEauthorrefmark{6}}
   \IEEEauthorblockA{\IEEEauthorrefmark{1}%
                     University of St. Thomas,
                     St. Paul, Minnesota, USA,
                    ande1298@stthomas.edu}
   \IEEEauthorblockA{\IEEEauthorrefmark{2}%
                     Penn State University,
                     University Park, Pennsylvania, USA,
                     abj5162@psu.edu}
                      \IEEEauthorblockA{\IEEEauthorrefmark{5}%
                      Carnegie Mellon University,
                      Pittsburgh, PA, USA,
                      gaurij@andrew.cmu.edu}
                      
   \IEEEauthorblockA{\IEEEauthorrefmark{3}%
                     Clemson University, 
                     Clemson, South Carolina, USA,
                     gmatthe@clemson.edu}
   \IEEEauthorblockA{\IEEEauthorrefmark{4}%
                     Worcester Polytechnic Institute, 
                     Worcester, Massachusetts, USA,
                     cdmayer@wpi.edu}
                                  \IEEEauthorblockA{\IEEEauthorrefmark{6}%
                      Rutgers University,
                      New Brunswick, NJ, USA,
                      emina.soljanin@rutgers.edu}        
                     
 }

\IEEEoverridecommandlockouts
\IEEEpubid{\makebox[\columnwidth]{\copyright2018 IEEE \hfill} \hspace{\columnsep}\makebox[\columnwidth]{ }}
\maketitle
\IEEEpubidadjcol

\begin{abstract}
  We consider storage systems in which $K$ files are stored over $N$ nodes. A node may be systematic for a particular file in the sense that access to it gives access to the file. Alternatively, a node may be coded, meaning that it gives access to a particular file only 
when combined with other nodes (which may be coded or systematic). Requests for file $f_k$ arrive at rate $\lambda_k$, and we are interested in the rate 
  that can be served by a particular system. In this paper, we determine the set of request arrival rates for the a $3$-file coded storage system. We also provide an algorithm to maximize the rate of requests served for file $K$ given \textbf{}$\lambda_1,\dots, \lambda_{K-1}$ in a general $K$-file case.

\end{abstract}


\section{Introduction}
\everymath{\small}

The explosive growth in the amount of data stored in the cloud calls for new techniques to make cloud infrastructure fast, reliable, and efficient. Moreover, applications that access this data from the cloud are becoming increasingly interactive. Thus, in addition to providing reliability against node failures, service providers must\rmv{have to} be able to serve a large number of users simultaneously.

Content files are typically replicated\rmv{files} at multiple nodes to cope with node failures. These replicas can also be used to serve a larger volume of users. To adapt to changes in popularity of content files, service providers can increase or decrease the number of replicates for each file, a strategy that has been widely used in content delivery networks \cite{borst2010distributed}. The use of erasure coding, instead of replication, to improve the availability of content is not yet fully understood. Using erasure codes has been shown to be effective in reducing the delay in accessing a file stored on multiple servers  \cite{joshi2014delay, shah2014mds, joshi2017efficient}.  However, only a few works have studied their use to store multiple files. Some recent works \cite{rawat_isit_2014} have proposed new classes of erasure codes to store multiple files that allow a file to be read from from disjoint sets of nodes. Other works \cite{swanand_allerton_2015,Simplex:AktasNS17} study the delay reduction achieved using these codes.

Besides  download latency, it has recently been recognized that another important metric for  the availability of  stored data is the service rate \cite{noori2016allocation,aktas2017service,joshi2017boosting}. Maximizing the service rate (or the throughput) of a distributed system helps support a large number of simultaneous system users. Rate-optimal strategies are also latency-optimal in high traffic. Thus, maximizing the service rate also reduces the latency experienced by users, particularly in highly contending scenarios. 

This paper is one of the first to analyze the \emph{service rate region} of a coded storage system. \rmv{The service capacity region achievable service rate region is the set of file request arrival rates\rmv{request rates for different files} that a system can support\rmv{can be supported}.} We consider distributed storage systems in which data for $K$ files is to be stored across $N$ nodes. A request for one of the files can be either  sent to a systematic node or to one of the repair groups. We seek to maximize such systems'\rmv{achievable service capacity region} service rate region, that is, the set of request arrival rates for the $K$ files that can be supported by a coded storage system. 

The problem addressed in this paper should not be confused with the related problem of caching and pre-fetching of popular content at edge devices \cite{Caching:AliN12}.
Caching benefits for the network are measured in reduction in the backhaul traffic it enables. Quality of service to the user measures include cache hit ratio and cache hit distance. Rather than with the backhaul, this paper is concerned with the access part of the network, namely, with potential service rate increase through work provided, jointly and possibly redundantly, by multiple network edge devices. Consequently,
instead of measuring e.g., content download performance by the likelihood of an individual cache hit or cache memory and bandwidth usage, we strive to ensure that multiple caches are jointly in possession of content and can deliver it fast to multiple simultaneous users.

In \cite{aktas2017service}, the \rmv{service capacity region}achievable service rate region was found for some common classes of codes, such as maximum-distance-separable (MDS) codes and simplex codes. That paper also determined the \rmv{capacity}service rate region when $K=2$, with  arbitrary numbers of systematic and coded nodes. We generalize this \rmv{service capacity}service rate region result from $K=2$ files to $K=3$ files and provide an algorithm to maximize the requests served for a given file with general $K$. 
The paper begins with preliminary notions given in Sec.~\ref{prelim}. Sec.~\ref{arb} addresses the general $K$ case where all nodes are coded, and Sec.~\ref{three} addresses the $K=3$ case. We return to the general case in Sec.~\ref{Kfiles_section}.
 

\section{Preliminaries} \label{prelim}

Suppose  files  $f_1, \ldots, f_K$ are stored across a  system that consists of $N$ nodes labeled $1,\ldots, N$.  For  $k \in [K]:= \left\{ 1, \dots, K \right\}$, there is a collection of minimal sets $R_{k1}, \dots,  R_{k \gamma_k} \subseteq [N]$ that each correspond to a set of nodes  that gives access to file $f_k$. Each such minimal set of nodes is called an {$f_k$--repair group}. 

\begin{example}
Fig. \ref{fig:RepairGroupExample} shows one possible way to store two files, $a$ and $b$, across four nodes. In this system, the $a$-repair groups are $\{1\},\{2\}$, and $\{3,4\}$. The $b$-repair groups are $\{4\},\{1,3\},$ and $\{2,3\}$.
\begin{figure}
\begin{center}
\begin{tikzpicture}[scale=1]
\draw [rotate around={0.:(-2.18,4.7)},fill=blue,fill opacity=.2] (-2.18,4.7) ellipse (0.5cm and 0.25cm);
\draw [rotate around={0.:(-0.86,4.7)},fill=blue,fill opacity=.2] (-0.86,4.7) ellipse (0.5cm and 0.25cm);
\draw [rotate around={0.:(0.48,4.7)},fill=blue,fill opacity=.2] (0.48,4.7) ellipse (0.5cm and 0.25cm);
\draw [rotate around={0.:(1.82,4.7)},fill=blue,fill opacity=.2] (1.82,4.7) ellipse (0.5cm and 0.25cm);
\draw (-2.2,4.7) node {$a$};
\draw (-2.2,4.1) node {$1$};
\draw (-.9,4.7) node {$a$};
\draw (-.9,4.1) node {$2$};
\draw (.5,4.7) node {$a+b$};
\draw (.5,4.1) node {$3$};
\draw (1.8,4.7) node{$b$};
\draw (1.8,4.1) node{$4$};
\end{tikzpicture}
\vspace{-.2in}
\end{center}
\label{fig:RepairGroupExample}
\caption{A possible way to store two files on $N=4$ nodes.}
\end{figure}
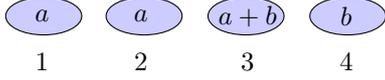
\end{example}

For  $(i,j) \in [\gamma_k]\times [N]$, define the function
 \begin{align}
  \delta_k(i,j) :=
  \begin{cases}
 1, & \text{if node $j$ is in the $f_k$--repair group $R_{ki}$}, \\
  0, & \text{else}.
  \end{cases}
  \end{align}

 Suppose that when a request for file  $f_k$ is received, that request is sent at random to  an $f_k$--repair group according to a  {splitting strategy} 
 with $\alpha_{ki} \geq 0$ denoting the fraction of requests sent to repair group $R_{ki}$, so that for each $k \in [K]$,
 \begin{equation}
 \sum_{i \in [\gamma_k]} \alpha_{ki}=1.
 \end{equation}
 Let  the {demand} for file $f_k$ be   $\lambda_k$, so  the arrival of requests for file $f_k$  to the storage system queue is Poisson with   rate $\lambda_k$, and let  $\boldsymbol{\lambda} = (\lambda_1, \ldots, \lambda_K)$ record the demand  for files $f_1, \ldots, f_K$. 
 
The average rate  that file requests arrive at  a storage system node depends  both on the splitting strategy for  file requests and  on the  demand  $\boldsymbol{\lambda}$.
 More precisely,  the average rate that file requests are received at node $j \in [N]$ is
 \begin{equation}
  \sum_{k \in [K]} \sum_{i \in [\gamma_k]} \alpha_{ik} \delta_k(i,j)\,\lambda_k.
  \end{equation}

Let $\mu_j$ denote the average rate of resolving received file requests at node $j$.    Whenever  demand  is such that  at least one  node $j$ of the storage system  receives requests at an average  rate in excess of its $\mu_j$,  the storage system queue will have a tendency to grow.  With this in mind, it is appropriate to call $\mu_j$ the {service rate} of node $j$. We will consider uniform systems for which $\mu_j=1$ for $j=1,\dots,N$.
 If, at   demand $\boldsymbol{\lambda}$,  there exists a splitting strategy under which no storage system node receives requests at a rate in excess of its service rate, then $\boldsymbol{\lambda}$  is said to be in the \rmv{service capacity region}achievable service rate region of the storage system.  
 More formally, the storage system's {\rmv{service capacity region}achievable service rate region $\mathcal{S}$ is the set of all $\boldsymbol{\lambda} \in \mathbb{R}^K_{\geq 0}$ such that  {there exists a splitting strategy} with
\begin{equation}  \label{SCcriteriaCapacity}
  \sum_{k \in [K]} \sum_{i \in [\gamma_k]} \alpha_{ki} \delta_k(i,j)\,\lambda_k \leq
  \mu_j,\quad \text{for all }j \in [N].
\end{equation}  
For any $\boldsymbol{\lambda}=(\lambda_1,\dots,\lambda_K)\in \mathbb{R}^K_{\geq 0}$, denote by  $\boldsymbol{\lambda}_{\widehat{k}}$ the $(K-1)$-tuple $(\lambda_1, \ldots, \lambda_{k-1},\lambda_{k+1}, \ldots \lambda_K)$, and for $x\in \mathbb{R}_{\geq 0}$ let $\boldsymbol{\lambda}_{\widehat{k}}\times \{x\}:=(\lambda_1, \ldots, \lambda_{k-1},x,\lambda_{k+1}, \ldots \lambda_K$). If $\boldsymbol{\lambda} \in \mathcal{S}$, then the same splitting strategy whose  existence  is guaranteed by (\ref{SCcriteriaCapacity})  is also sufficient to give $\boldsymbol{\lambda}'\in \mathcal{S}$ for every $\boldsymbol{\lambda}'$  satisfying  for all $k \in [K]$, $\lambda'_k \leq \lambda_k$.   
 Thus, 
 given any pair  $\boldsymbol{\lambda}_{\widehat{k}} \times \{0\} \in \mathcal{S}$ and    $\boldsymbol{\lambda}_{\widehat{k}} \times \{\lambda_k\}  \in \mathcal{S}$, 
 the entire interval $\boldsymbol{\lambda}_{\widehat{k}}\times [0,\lambda_k]$ is  in $\mathcal{S}$.  Moreover, for any storage system (regardless of its coding), if $\boldsymbol{\lambda}$ 
  is such that the demand for any file $f_k$ is in excess of $N \cdot \max_{j \in [N]} \{\mu_j\}$, 
  then under  all possible assignment strategies (\ref{SCcriteriaCapacity}) is violated for at least one node $j\rmv{ \in [N]}$, and so 
  $\boldsymbol{\lambda}_{\widehat{k}}\times \{x\}$ is not in $\mathcal{S}$
   for any $x>N\cdot{\max_{j \in [N]} \{\mu_j\}}$ and $\boldsymbol{\lambda}_{\widehat{k}} \in \mathbb{R}^{K-1}_{\geq 0}$. In this way, $\mathcal{S}$ is a non-empty, closed, and bounded subset of $\mathbb{R}^{K}_{\geq 0}$.  Therefore, given any $\boldsymbol{\lambda}_{\widehat{k}}\times\{0\} \in \mathcal{S}$, there exists a maximal value of $\lambda_k$ such that  $\boldsymbol{\lambda}_{\widehat{k}}\times [0, \lambda_k] \subset \mathcal{S}$ and $\boldsymbol{\lambda}_{\widehat{k}}\times  \{\lambda_k'\} \not\in \mathcal{S}$  for any $\lambda_k' >\lambda_k$.  When $k=K$, we call this maximal value $L(\boldsymbol{\lambda}_{\widehat{K}})$. \rmv{
 $L : \mathbb{R}^{K-1}_{\geq0}\rightarrow\mathbb{R}_{\geq0}$ defined by 
   $$
   \boldsymbol{\lambda} \mapsto 
  \begin{cases}
                                   L\Big( (\lambda_1, \ldots, \lambda_{K-1})\Big), & \text{if $(\lambda_1, \ldots, \lambda_{K-1},0)\in \mathcal{S}$} \\
                                   0, & \text{else.}
  \end{cases}
$$ 
\rmv{
$$
\begin{array}{llll} 
L : &\mathbb{R}^{K-1}_{\geq0} &\rightarrow &\mathbb{R}_{\geq0}\\
   &\boldsymbol{\lambda} &\mapsto 
  &\begin{cases}
                                   L\Big( \boldsymbol{\lambda} \Big), & \text{if $(\lambda_1, \ldots, \lambda_{K-1},0)\in \mathcal{S}$} \\
                                   0, & \text{else.}
  \end{cases}
  \end{array}
$$ }
\rmv{ \begin{align}
  L :\quad & \mathbb{R}^{K-1}_{\geq0}\rightarrow\mathbb{R}_{\geq0}\nonumber\\
  & (\lambda_1, \ldots, \lambda_{K-1}) \mapsto 
  \begin{cases}
                                   L\Big( (\lambda_1, \ldots, \lambda_{K-1})\Big), & \text{if $(\lambda_1, \ldots, \lambda_{K-1},0)\in \mathcal{S}$} \\
                                   0, & \text{else.}
  \end{cases}
 \end{align}  }}
 In this notation, 
 %
the \rmv{service capacity region}service rate region of any storage system can be described as: 
  \begin{equation}
 \label{eq:SwrtL}
 \mathcal{S} =  \{\boldsymbol{\lambda}_{\widehat{K}}\times [0,L(\boldsymbol{\lambda}_{\widehat{K}})]:(\lambda_1, \ldots, \lambda_{K-1}, 0) \in \mathcal{S}\}.
\end{equation} 
\vspace{-.05in}
\begin{example}
Three examples of how two files, $a$ and $b$ may be stored across three nodes are shown on the below on the left. The resulting service \rmv{capacity regions} service rate regions for each system are shown below on the right.
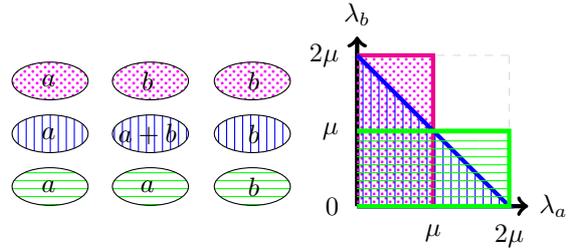
\begin{center}
\vspace{-.05in}
\begin{tikzpicture}[scale=1]
\draw [rotate around={0.:(-2.18,4.7)},pattern=crosshatch dots,pattern color=magenta] (-2.18,4.7) ellipse (0.5cm and 0.25cm);
\draw [rotate around={0.:(-0.86,4.7)},pattern=crosshatch dots,pattern color=magenta] (-0.86,4.7) ellipse (0.5cm and 0.25cm);
\draw [rotate around={0.:(0.48,4.7)},pattern=crosshatch dots,pattern color=magenta] (0.48,4.7) ellipse (0.5cm and 0.25cm);
\draw (-2.2,4.7) node {$a$};
\draw (-.9,4.7) node {$b$};
\draw (.5,4.7) node {$b$};

\draw [rotate around={0.:(-2.18,3.7)},pattern=vertical lines, pattern color=blue] (-2.18,4) ellipse (0.5cm and 0.25cm);
\draw [rotate around={0.:(-0.86,3.7)},pattern=vertical lines, pattern color=blue] (-0.86,4) ellipse (0.5cm and 0.25cm);
\draw [rotate around={0.:(0.48,3.7)},pattern=vertical lines, pattern color=blue] (0.48,4) ellipse (0.5cm and 0.25cm);
\draw (-2.2,4) node {$a$};
\draw (-.9,4) node {$a+b$};
\draw (.5,4) node {$b$};

\draw [rotate around={0.:(-2.18,2.7)},pattern= horizontal lines,pattern color=green] (-2.18,3.3) ellipse (0.5cm and 0.25cm);
\draw [rotate around={0.:(-0.86,2.7)},pattern= horizontal  lines,pattern color=green] (-0.86,3.3) ellipse (0.5cm and 0.25cm);
\draw [rotate around={0.:(0.48,2.7)},pattern= horizontal  lines,pattern color=green] (0.48,3.3) ellipse (0.5cm and 0.25cm);
\draw (-2.2,3.3) node {$a$};
\draw (-.9,3.3) node {$a$};
\draw (.5,3.3) node {$b$};
\node (spacing) at (.5,2.5){};
\end{tikzpicture}
\begin{tikzpicture}[ultra thick,scale=1]
\draw[help lines, color=gray!30, dashed] (-.01,-.01) grid (2.1,2.1);
\draw[->,ultra thick,black] (0,0)--(2.25,0) node[right]{$\lambda_a$};
\draw[->,ultra thick,black] (0,0)--(0,2.25) node[above]{$\lambda_b$};
\draw[ultra thick,magenta,pattern=crosshatch dots,pattern color=magenta] (0,2)--(1,2)--(1,0)--(0,0);
\draw[ultra thick,blue,pattern=vertical lines, pattern color=blue] (0,2)--(1,1)--(2,0)--(0,0);
\draw[ultra thick,color=green,pattern= horizontal lines,pattern color=green] (0,1)--(2,1)--(2,0)--(0,0);

\node[circle,scale=.5,label=left:$0$] (O) at (0,0) {};
\node[circle,scale=.5,label=below:$\mu$] (x1) at (1,0) {};
\node[circle,scale=.5,label=below:$2\mu$] (x2) at (2,0) {};
\node[circle,scale=.5,label=left:$\mu$] (y1) at (0,1) {};
\node[circle,scale=.5,label=left:$2\mu$] (y2) at (0,2) {};
\end{tikzpicture}
\end{center}
\end{example}
Coding schemes that  use a mixture of replication and MDS coding are not conventional.  However, if the service rate region is used as a performance metric, then a combination of coded and systematic nodes has been shown to be beneficial \cite{swanand_allerton_2015,aktas2017service}. In this paper, we consider storage systems for $K$ files whose coded nodes  satisfy  the following three conditions:
\begin{enumerate}
 \item Each $K$--subset of coded nodes forms an $f_k$--repair group for every $k \in [K]$.
 \item No subset of $k<K$ coded nodes forms an $f_k$--repair group, for any $k \in [K]$.
 \item With  addition of  systematic nodes for any $n$ distinct files (naturally, $n<K$) every  $(K-n)$--subset of coded nodes from the core  completes these systematic nodes to form an $f_k$--repair group for every $k \in [K]$.
\end{enumerate}
We say that such a system has an MDS core. We consider situations with uniform node capacities $\mu = \mu_1 = \cdots = \mu_N$. 

\rmv{In this scenario, files may have different levels of popularity. Some files may contain hot data, which means it is accessed frequently, while other files may not be downloaded as often. We are interested in  the availability of the stored data to users. A measure of this is the service capacity region, which is the space of download request rates for which the system is stable. This has been considered in \cite{noori2016allocation}. In \cite{aktas2017service}, the service capacity region is fully described for MDS $2$-file cores in storage systems with arbitrarily many systematic nodes for both files. 

In this paper, we consider two situations, both with uniform node capacities. Section \ref{arb} addresses the $K$ case where all nodes are coded, and Section \ref{three} addresses the $K=3$ file case. We return to the general case in Section \ref{Kfiles_section}. The paper begins with preliminary notions are given in Section \ref{prelim} and concludes with Section \ref{conclusion}.}

For convenience, we use $C$ to denote the number of coded nodes in such a  core. When  systematic nodes are also present, we use $N_k$ to denote the number of systematic nodes for file $f_k$.  In this way, the total number of nodes in a storage system for $K$ files that has an MDS core is $
N = C + \sum_{k=1}^K N_k$.
 

\section{All coded nodes} \label{arb}

We begin by considering an MDS $K$-file core where there are no systematic nodes in the system. In this situation, all nodes form a repair group for each file, and $K$ 
nodes are required to recover any file. 

\begin{theorem}
\label{allcodedK}
Assume $N_1= \cdots = N_K = 0$.  If  there are $C > K - 1$ coded nodes, then the \rmv{service capacity region}achievable service rate region $\mathcal{S}$ is the set of all $\boldsymbol{\lambda}$ with $\sum_{i=1}^K \lambda_i\leq \frac{C}{K}\mu$, and so   $L(\lambda_1, \ldots, \lambda_{K-1}) =\frac{C}{K}\mu-\sum_{i = 1}^{K - 1} \lambda_i$.
If there are $C \leq K - 1$ coded nodes, then  \rmv{the service capacity region}$\mathcal{S}$ is the point $(0,\ldots,0).$
\end{theorem}
\begin{proof}
If $C \leq K - 1$, then no file can be recovered and the \rmv{service capacity region} service rate region is the point $(0,\ldots,0)$. 

Assume $C > K - 1$. Note that since every repair group requires $K$ nodes, the total demand that can be served is bounded above by $\frac{C\mu}{K}$. For each file, there are a total of $\binom{C}{K}$ repair groups, and each node is in $\binom{C-1}{K-1}$ repair groups. By sending demand $\frac{\mu}{\binom{C-1}{K-1}}$ \rmv{requests} to each repair group, requests to each node occur at the service rate and the system can serve demand $\frac{\binom{C}{K}}{\binom{C-1}{K-1}}\mu=\frac{C}{K}\mu$ \rmv{requests}. Since this demand\rmv{these requests} can be for any file, the \rmv{service capacity} service rate region is $\sum_{i=1}^K \lambda_i\leq \frac{C}{K}\mu$. Therefore, the maximum achievable $\lambda_K$ is 
\vspace{-.1in}
\[
\lambda_K = L(\lambda_1, \ldots, \lambda_{K - 1}) = \frac{C}{K}\mu - \sum_{i = 1}^{K - 1} \lambda_i. 
\]
\vspace{-.1in}
\end{proof}
The two file case is considered in \cite{aktas2017service}. The situation becomes increasingly complex depending on the number of files $K$ in the system. In the next section, we consider  $K=3$. 

\section{Three files} \label{three}

In this section, we consider the \rmv{service capacity region}service rate region of storage systems for $3$ files with MDS cores. As a corollary to Theorem \ref{allcodedK}, we obtain the service rate region for the case when there are no systematic nodes, which is represented in Fig. \ref{allcodedfigure}. Note that when the demand for one file is zero, then this may be considered a system with only two files. For example, if $\lambda_3 = 0$, then the maximum achievable $\lambda_2$ is $\lambda_2 = \frac{C}{3}\mu -\lambda_1$, which is the region shaded in Fig. \ref{allcodedfigure}.

\begin{figure}[h]
\begin{center}
\vspace{-.15in}
\begin{tikzpicture}
[scale=.9, vertices/.style={draw, fill=black, circle, inner sep=0.75pt}]
\node[vertices, label=below:{$\frac{C}{2}\mu $}] (a) at (0,0) {};
\node[vertices, red,fill=red,label=left:{$0$}] (c) at (-1.5,0) {};
\node[vertices, label=left:{$\frac{C}{2}\mu $}] (d) at (-1.5,1.5) {};
\draw (a)--(c);
\draw (a)--(d);
\draw (d)--(c);
    \draw[->,thin,black] (-1.5,0)--(.2,0) node[right]{$\lambda_1$};
    \draw[->,thin,black] (-1.5,0)--(-1.5,1.7) node[above]{$\lambda_2$}; 
\end{tikzpicture}
\hspace{1cm}
\begin{tikzpicture}[scale=.9, vertices/.style={draw, circle, inner sep=0.75pt}]
\draw[->,thin,black] (-.875,.5)--(.15,-.075) node[right]{$\lambda_1$};
\draw[->,thin,black] (-.875,.5)--(-1.875,-.1) node[left]{$\lambda_3$};
\draw[->,thin,black] (-.875,.5)--(-.875,1.7) node[above]{$\lambda_2$};
\node[ fill=black,vertices, label=below:{$\frac{C}{3}\mu$}] (a) at (0,0) {};
\node[fill=black,vertices, label=below right:{$\frac{C}{3}\mu $}] (b) at (-1.75,0) {};
\node[ fill=black,scale=1,red,fill=red,vertices, label=left:{$0$}] (c) at (-.875,0.5) {};
\node[fill=black,vertices, label=left:{$\frac{C}{3}\mu$}] (d) at (-.875,1.5) {};
\foreach \to/\from in {b/c,c/d}
	\draw [dashed] (\to)--(\from);
\draw[ultra thick, dashed,blue](a)--(c);
\fill[green,fill opacity=0.7] (-.875,.5) --(0,0)--(-.875,1.5)--(-.875,.5)--cycle;
\foreach \to/\from in {a/d,b/a,b/d}
	\draw [-] (\to)--(\from);
\end{tikzpicture}
\vspace{-.15in}
\end{center}
\caption{\rmv{Service capacity region} Achievable service rate regions of all-coded-node systems with $2$ files (left) or $3$ files (right).}
\label{allcodedfigure}
\end{figure}
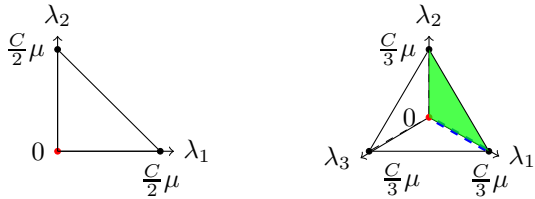

We now consider storage systems that have both coded nodes and systematic nodes. 
Suppose that a coded storage system has $C$ coded nodes and $N_i$ systematic file $f_i$ nodes, $i=1,2,3$. Note that a systematic repair node may be in a repair group with a single node (serving requests for the file it stores) or three nodes (serving requests for any other file).  Any repair group using a coded node contains three nodes.  
For $i=1,2$, if $r_i\leq N_i\mu$ requests for file $f_i$ are served using systematic $f_i$ nodes (and any other demand for file $f_i$ is served using a repair group of three nodes), then the total demand that can be served is bounded above by 
$$D:=r_1+r_2+\frac{(N_1\mu-r_1)+(N_2\mu-r_2)+C\mu}{3}+N_3\mu.$$ 
Given demand $\lambda_1$ for file $f_1$ and $\lambda_2$ for file $f_2$, the rate of requests that may be served for file $f_3$ is bounded above by $\max\{D-\lambda_1-\lambda_2,0\}.$ This is maximized when $r_i=\min\{\lambda_i,N_i\mu\}$ for $i=1,2$. The splitting strategy in the proof of the following theorem meets this bound.
\begin{theorem}
\label{3files}
Assume there are $N_1$, $N_2$, and $N_3$ systematic nodes for files $f_1$, $f_2$, and $f_3$, respectively, and $C$ coded nodes. Assume $\lambda_1 + \lambda_2 \leq \mu N_1 + \mu N_2 + \frac{C}{3} \mu$ and $C \geq \max \left(3, N_{1} - \frac{\lambda_1}{\mu}, N_2 - \frac{\lambda_2}{\mu} \right)$. 
Then $\mathcal{S}$ has  
$L(\lambda_1, \lambda_2)=$ 
\begin{align*}
\begin{cases}
(\frac{C}{3}+\frac{N_1}{3}+\frac{N_2}{3}+N_3)\mu  -\frac{\lambda_{1}}{3} - \frac{\lambda_{2}}{3}, &\hspace{-.05in} 0 \leq \frac{\lambda_i}{\mu} \leq N_i,\,i=1,2 \\
(\frac{C}{3}+N_1 +\frac{N_2}{3}+ N_3)\mu- \lambda_1- \frac{\lambda_2}{3}  , &\hspace{-.05in} N_1 < \frac{\lambda_1}{\mu} \leq N_1 +  N_2 + \frac{C}{3},\\ &\hspace{-.05in}0 \leq \frac{\lambda_2}{\mu} \leq N_2 \\
(\frac{C}{3}+ \frac{N_1}{3}+N_2 + N_3)\mu - \frac{\lambda_1}{3}  - \lambda_2, &\hspace{-.05in} 0 \leq \frac{\lambda_1}{\mu} \leq  N_1,\\ &\hspace{-.05in} N_2 < \frac{\lambda_2}{\mu} \leq  N_1 +N_2 + \frac{C}{3} \\
(\frac{C}{3}+ N_1 + N_2 + N_3)\mu - \lambda_1 -\lambda_2, & \hspace{-.05in} N_1 < \frac{\lambda_1}{\mu} \leq  N_1 + N_2 + \frac{C}{3},\\ &\hspace{-.05in} N_2 < \frac{\lambda_2}{\mu} \leq  N_1 + N_2 + \frac{C}{3}
\end{cases}
\end{align*}

\end{theorem}

\begin{proof}
Consider a system with $N_1$, $N_2$, and $N_3$ systematic nodes for files $f_1$, $f_2$, and $f_3$ and $C$ coded nodes. \\


\textbf{Step 1:  Send requests to systematic nodes at the service rate to serve demand for files $f_1$ and $f_2$, as possible.  If any $f_i$ ($i=1,2$) systematic nodes remain available, distribute remaining file $f_i$ demand uniformly across those nodes.} 
\begin{example} \label{step1} Consider a $3$-file system with $N_1=3$, $N_2=1$, $N_3=1$, and $C=3$.
\begin{center}
\begin{tikzpicture}[scale=1]
\draw [rotate around={0.:(-2,4.7)}] (-2,4.7) ellipse (0.5cm and 0.25cm);
\draw [rotate around={0.:(-0.9,4.7)}] (-0.9,4.7) ellipse (0.5cm and 0.25cm);
\draw [rotate around={0.:(0.2,4.7)}] (0.2,4.7) ellipse (0.5cm and 0.25cm);
\draw [rotate around={0.:(1.3,4.7)}] (1.3,4.7) ellipse (0.5cm and 0.25cm);
\draw [rotate around={0.:(2.4,4.7)}] (2.4,4.7) ellipse (0.5cm and 0.25cm);
\draw [rotate around={0.:(3.5,4.7)}] (3.5,4.7) ellipse (0.5cm and 0.25cm);
\draw [rotate around={0.:(4.6,4.7)}] (4.6,4.7) ellipse (0.5cm and 0.25cm);
\draw [rotate around={0.:(5.7,4.7)}] (5.7,4.7) ellipse (0.5cm and 0.25cm);
\draw (-2,4.65) node {$f_1$};
\draw (-.9,4.65) node {$f_1$};
\draw (.2,4.65) node {$f_1$};
\draw (1.3,4.65) node {$f_2$};
\draw (2.4,4.65) node {$f_3$};
\draw (3.5,4.65) node {$c$};
\draw (4.6,4.65) node {$c$};
\draw (5.7,4.65) node {$c$};
\end{tikzpicture}
\end{center}
If $\lambda_1=\frac{3}{2}\mu$ and $\lambda_2=2\mu$ then $\mu$ requests for $f_1$ will be served by one of the $f_1$ systematic nodes, and the remaining $\frac{1}{2}\mu$ requests for $f_1$ will be split between the other $2$ systematic nodes. Also, $\mu$ requests for $f_2$ will be served by the $f_2$ systematic node. After Step 1, the remaining demand for $f_1$ is $0$ and the remaining demand for file $f_2$ is $\mu$. In the system, there are now two systematic $f_1$ nodes that can handle an additional $\frac{3}{4}\mu$ requests as well as one systematic $f_3$ node and three coded nodes each with available service rate $\mu$.
\begin{center}
\begin{tikzpicture}[scale=1]
\tikzset{
    partial ellipse/.style args={#1:#2:#3}{
        insert path={+ (#1:#3) arc (#1:#2:#3)}
    }
}

\draw [fill=gray,rotate around={0.:(-2,4.7)}] (-2,4.7) ellipse (0.5cm and 0.25cm);
\draw [rotate around={0.:(-0.9,4.7)}] (-0.9,4.7) ellipse (0.5cm and 0.25cm);
\draw[fill=gray,rotate around={0.:(-0.9,4.7)}] (-0.9,4.7) [partial ellipse=200:340:.5cm and .25cm];
\draw [rotate around={0.:(0.2,4.7)}] (0.2,4.7) ellipse (0.5cm and 0.25cm);
\draw[fill=gray,rotate around={0.:(0.2,4.7)}] (0.2,4.7) [partial ellipse=200:340:.5cm and .25cm];
\draw [fill=gray,pattern=checkerboard light gray,rotate around={0.:(1.3,4.7)}] (1.3,4.7) ellipse (0.5cm and 0.25cm);
\draw [rotate around={0.:(2.4,4.7)}] (2.4,4.7) ellipse (0.5cm and 0.25cm);
\draw [rotate around={0.:(3.5,4.7)}] (3.5,4.7) ellipse (0.5cm and 0.25cm);
\draw [rotate around={0.:(4.6,4.7)}] (4.6,4.7) ellipse (0.5cm and 0.25cm);
\draw [rotate around={0.:(5.7,4.7)}] (5.7,4.7) ellipse (0.5cm and 0.25cm);
\draw (-2,4.65) node {$f_1$};
\draw (-.9,4.65) node {$f_1$};
\draw (.2,4.65) node {$f_1$};
\draw (1.3,4.65) node {$f_2$};
\draw (2.4,4.65) node {$f_3$};
\draw (3.5,4.65) node {$c$};
\draw (4.6,4.65) node {$c$};
\draw (5.7,4.65) node {$c$};
\end{tikzpicture}
\end{center}
\end{example}
At the end of Step 1, if $\lambda_i \leq \mu N_i$ for $i = 1$ or $2$, then there will be $N_i'=N_i-\lfloor\frac{\lambda_i}{\mu}\rfloor$ systematic nodes remaining available for $f_i$, each with service rate reduced to $\mu_i'=\mu-\frac{\lambda_i-\lfloor\frac{\lambda_i}{\mu}\rfloor\cdot\mu}{N_i'}$.
Since $\lambda_i \leq \mu N_i$, the remaining demand for file $f_i$ is $\lambda_i'=0$.

If $\lambda_i \geq \mu N_i$ for $i = 1$ or $2$, we exhaust every $f_i$ systematic node. 
The remaining demand for file $f_i$ is then $\lambda_i' = \lambda_i - \mu N_{i}$, and $N_i'=0$ systematic $f_i$ nodes remain.

\textbf{Step 2: Serve any remaining demand for files $f_1$ and $f_2$. Finally, serve demand for file $f_3$.}
\begin{example}\label{step2}
Consider the system in Example \ref{step1}. In Step 2 we want to serve the remaining requests for file $f_2$ in a way that maximizes the requests that can be handled for $f_3$. In particular, we will reserve the use of systematic $f_3$ nodes for accessing file $f_3$.
Note that there are $2\cdot\binom{3}{2}=6$ repair groups for file $f_2$ that involve one systematic $f_1$ node and  two coded nodes. If we send $\frac{\mu}{6}$ requests for file $f_2$ to each of these repair groups, then all the requests for file $f_2$ are served, each $f_1$ systematic node can serve  $\frac{\mu}{4}$ more requests (as each $f_1$ node is in $3$ repair groups) and each coded node can serve $\frac{\mu}{3}$ more requests (as each coded node is in $4$ repair groups).

\begin{center}
\begin{tikzpicture}[scale=1]
\tikzset{
    partial ellipse/.style args={#1:#2:#3}{
        insert path={+ (#1:#3) arc (#1:#2:#3)}
    }
}

\draw [fill=gray,rotate around={0.:(-2,4.7)}] (-2,4.7) ellipse (0.5cm and 0.25cm);
\draw [rotate around={0.:(-0.9,4.7)}] (-0.9,4.7) ellipse (0.5cm and 0.25cm);
\draw[pattern=checkerboard light gray,rotate around={0.:(-0.9,4.7)}] (-0.9,4.7) [partial ellipse=160:380:.5cm and .25cm];
\draw[fill=gray,rotate around={0.:(-0.9,4.7)}] (-0.9,4.7) [partial ellipse=200:340:.5cm and .25cm];
\draw [rotate around={0.:(0.2,4.7)}] (0.2,4.7) ellipse (0.5cm and 0.25cm);
\draw[pattern=checkerboard light gray,rotate around={0.:(0.2,4.7)}] (0.2,4.7) [partial ellipse=160:380:.5cm and .25cm];
\draw[fill=gray,rotate around={0.:(0.2,4.7)}] (0.2,4.7) [partial ellipse=200:340:.5cm and .25cm];
\draw [pattern=checkerboard light gray,rotate around={0.:(1.3,4.7)}] (1.3,4.7) ellipse (0.5cm and 0.25cm);
\draw [rotate around={0.:(2.4,4.7)}] (2.4,4.7) ellipse (0.5cm and 0.25cm);
\draw [rotate around={0.:(3.5,4.7)}] (3.5,4.7) ellipse (0.5cm and 0.25cm);
\draw[pattern=checkerboard light gray,rotate around={0.:(3.5,4.7)}] (3.5,4.7) [partial ellipse=165:375:.5cm and .25cm];
\draw [rotate around={0.:(4.6,4.7)}] (4.6,4.7) ellipse (0.5cm and 0.25cm);
\draw[pattern=checkerboard light gray,rotate around={0.:(4.6,4.7)}] (4.6,4.7) [partial ellipse=165:375:.5cm and .25cm];
\draw [rotate around={0.:(5.7,4.7)}] (5.7,4.7) ellipse (0.5cm and 0.25cm);
\draw[pattern=checkerboard light gray,rotate around={0.:(5.7,4.7)}] (5.7,4.7) [partial ellipse=165:375:.5cm and .25cm];
\draw (-2,4.65) node {$f_1$};
\draw (-.9,4.65) node {$f_1$};
\draw (.2,4.65) node {$f_1$};
\draw (1.3,4.65) node {$f_2$};
\draw (2.4,4.65) node {$f_3$};
\draw (3.5,4.65) node {$c$};
\draw (4.6,4.65) node {$c$};
\draw (5.7,4.65) node {$c$};
\end{tikzpicture}
\end{center}

Finally, requests for $f_3$ may be served. Sending $\frac{\mu}{12}$ requests to each of the $6$ repair groups with one $f_1$ node and two coded nodes exhausts each $f_1$ node and each coded node. The full service rate of the systematic $f_3$ node may also be used to serve requests for $f_3$. Thus a total of $6\cdot\frac{\mu}{12}+\mu=\frac{3}{2}\mu$ requests for $f_3$ may be served.
\begin{center}
\begin{tikzpicture}[scale=1]
\tikzset{
    partial ellipse/.style args={#1:#2:#3}{
        insert path={+ (#1:#3) arc (#1:#2:#3)}
    }
}

\draw [fill=gray,rotate around={0.:(-2,4.7)}] (-2,4.7) ellipse (0.5cm and 0.25cm);
\draw [pattern=horizontal lines light blue,rotate around={0.:(-0.9,4.7)}] (-0.9,4.7) ellipse (0.5cm and 0.25cm);
\draw[pattern=horizontal lines light blue,pattern=checkerboard light gray,rotate around={0.:(-0.9,4.7)}] (-0.9,4.7) [partial ellipse=160:380:.5cm and .25cm];
\draw[fill=gray,rotate around={0.:(-0.9,4.7)}] (-0.9,4.7) [partial ellipse=200:340:.5cm and .25cm];
\draw [pattern=horizontal lines light blue,rotate around={0.:(0.2,4.7)}] (0.2,4.7) ellipse (0.5cm and 0.25cm);
\draw[pattern=checkerboard light gray,rotate around={0.:(0.2,4.7)}] (0.2,4.7) [partial ellipse=160:380:.5cm and .25cm];
\draw[fill=gray,rotate around={0.:(0.2,4.7)}] (0.2,4.7) [partial ellipse=200:340:.5cm and .25cm];
\draw [pattern=checkerboard light gray,rotate around={0.:(1.3,4.7)}] (1.3,4.7) ellipse (0.5cm and 0.25cm);
\draw [pattern=horizontal lines light blue,rotate around={0.:(2.4,4.7)}] (2.4,4.7) ellipse (0.5cm and 0.25cm);
\draw [pattern=horizontal lines light blue,rotate around={0.:(3.5,4.7)}] (3.5,4.7) ellipse (0.5cm and 0.25cm);
\draw[pattern=checkerboard light gray,rotate around={0.:(3.5,4.7)}] (3.5,4.7) [partial ellipse=165:375:.5cm and .25cm];
\draw [pattern=horizontal lines light blue,rotate around={0.:(4.6,4.7)}] (4.6,4.7) ellipse (0.5cm and 0.25cm);
\draw[pattern=checkerboard light gray,rotate around={0.:(4.6,4.7)}] (4.6,4.7) [partial ellipse=165:375:.5cm and .25cm];
\draw [pattern=horizontal lines light blue,rotate around={0.:(5.7,4.7)}] (5.7,4.7) ellipse (0.5cm and 0.25cm);
\draw[pattern=checkerboard light gray,rotate around={0.:(5.7,4.7)}] (5.7,4.7) [partial ellipse=165:375:.5cm and .25cm];
\draw (-2,4.65) node {$f_1$};
\draw (-.9,4.65) node {$f_1$};
\draw (.2,4.65) node {$f_1$};
\draw (1.3,4.65) node {$f_2$};
\draw (2.4,4.65) node {$f_3$};
\draw (3.5,4.65) node {$c$};
\draw (4.6,4.65) node {$c$};
\draw (5.7,4.65) node {$c$};
\end{tikzpicture}
\end{center}
\end{example}
How requests are served in Step 2 depends on the demand and number of systematic nodes for files $f_1$ and $f_2$. Let $\lambda$ be the total demand for files $f_1$ and $f_2$ that remains after Step 1; that is, $\lambda = \lambda_1' + \lambda_2'$.


\textit{\textbf{Case 1} ($0 \leq \lambda_1 \leq \mu N_1$, $0 \leq \lambda_2 \leq \mu N_2$):} 
In this case, $\lambda=0$,  so all  available system resources  may be used to serve  $f_3$ demand.  The full service rate of file $f_3$ systematic nodes may be used,  serving  demand $\mu N_3$ for file $f_3$.
Let $\sigma$ be a permutation on $\{1, 2\}$ such that $\frac{\mu_{\sigma(1)}'}{N_{\sigma(2)}'}\leq \frac{\mu_{\sigma(2)}'}{N_{\sigma(1)}'}$. 

There are $N_1'N_2'C$ $f_3$--repair groups with a systematic node 
for each of $f_1$ and $f_2$, and one coded node. 
Recall, $C \geq \max \left(3, N_{1} - \frac{\lambda_1}{\mu}, N_2 - \frac{\lambda_2}{\mu} \right)$. 
Since $C \geq N_{\sigma(1)} - \frac{\lambda_{\sigma(1)}}{\mu}$, 
$$\mu_{\sigma(1)}'N_{\sigma(1)}' =  \mu  \left(N_{\sigma(1)} -  \frac{\lambda_{\sigma(1)}}{\mu} \right)$$
demand for  $f_3$ can be served by sending $\frac{\mu_{\sigma(1)}'}{N_{\sigma(2)}'C}$ demand to each of these repair groups. 
The service rate of each  $f_{\sigma(1)}$ is reduced to $0$, while  $f_{\sigma(2)}$ systematic nodes have $\mu_{\sigma(2)}''=\mu_{\sigma(2)}'-\frac{\mu_{\sigma(1)}'}{N_{\sigma(2)}'C}N_{\sigma(1)}'C=\mu_{\sigma(2)}'-\frac{\mu_{\sigma(1)}'}{N_{\sigma(2)}'}N_{\sigma(1)}'$, and coded nodes have $\mu'_C=\mu-\frac{\mu_{\sigma(1)}'}{N_{\sigma(2)}'C}N_{\sigma(1)}'N_{\sigma(2)}'=\mu-\frac{\mu_{\sigma(1)}'}{C}N_{\sigma(1)}'$.




There are $N'_{\sigma(2)}\binom{C}{2}$ $f_3$--repair groups with one of the remaining systematic file $f_{\sigma(2)}$ nodes and $2$ coded nodes. Since $C \geq N_{\sigma(2)} - \frac{\lambda_{\sigma(2)}}{\mu}$, similarly to before, we can serve $$\mu''_{\sigma(2)}N_{\sigma(2)}' = \mu \left( \left(N_{\sigma(2)} - \frac{\lambda_{\sigma(2)}}{\mu}\right) - \left(N_{\sigma(1)} - \frac{\lambda_{\sigma(1)}}{\mu}\right)\right) $$ demand for file $f_3$ by sending demand equally to each of these $f_3$--repair groups. 
Each coded node has remaining service rate $\mu''_C=\mu'_C-\frac{\mu''_{\sigma(2)}}{\binom{C}{2}}(C-1)N_{\sigma(2)}',$ and no systematic $f_1, f_2$ nodes remain available. 

Since $C\geq 3$, as in the case  in Theorem \ref{allcodedK} with $C$ coded nodes and no systematic nodes, 
the service rate $\mu_C''$ of these coded nodes can be used to serve $\frac{C}{3}\mu''_C$ demand for file $f_3$.


Thus, the maximum achievable $\lambda_3$ is 
$L(\lambda_1, \lambda_2)$
\begin{align*}
=& \frac{C}{3}\mu''_C+ \mu''_{\sigma(2)}N_{\sigma(2)}'+ \mu_{\sigma(1)}'N_{\sigma(1)}' +\mu N_3\\
=& \frac{1}{3}\left(C \mu + \mu N_{\sigma(2)} -\lambda_{\sigma(2)} + \mu N_{\sigma(1)} - \lambda_{\sigma(1)} \right) +\mu N_3.
\end{align*}

Similar arguments can be used for \textbf{Case 2}: $\mu N_1 < \lambda_1 \leq \mu N_1 + \mu N_2 + \frac{C}{3}\mu$, $0 \leq \lambda_2 \leq \mu N_2$ and \textbf{Case 3}: $0 \leq \lambda_1 \leq \mu N_1$, $\mu N_2 < \lambda_2 \leq \mu N_1 + \mu N_2 + \frac{C}{3}\mu$ (see Example \ref{step2}).

\textit{\textbf{Case 4} ($\mu N_1 < \lambda_1 \leq \mu N_1 + \mu N_2 + \frac{C}{3}\mu$, $\mu N_2 < \lambda_2 \leq \mu N_1 + \mu N_2 + \frac{C}{3}\mu$):}
In this case, 
all available repair groups consist entirely of coded nodes. 
Since  demand $\mu N_i$  for file $f_i$ ($i=1,2$) was satisfied in Step 1,
the remaining total demand for files $f_1$ and $f_2$ is $\lambda < \frac{C}{3}\mu$.
Since $C \geq 3$, this can be served by sending demand equally to every  coded repair group. The coded nodes' remaining ability to service can be used for file $f_3$.
Thus, the maximum achievable $\lambda_3$ is 
\begin{align*}
L(\lambda_1, \lambda_2)&= \frac{C}{3}\mu - \lambda + \mu N_3\\
&= \frac{C}{3}\mu - (\lambda_1 - \mu N_1 + \lambda_2 - \mu N_2) + \mu N_3.
\vspace*{-3ex}
\end{align*}
\end{proof}

Note that $L(\lambda_1,\lambda_2)$ can be found for  systems with $C<3$ coded nodes in a similar way. When $C<3$, all repair groups must contain systematic nodes for at least $3-C$ distinct files.

\section{MDS $K$-file cores} \label{Kfiles_section}

Theorem \ref{3files} may be generalized to provide an algorithm for maximizing $\lambda_k$ for the general $K$-file case. Assume we have an MDS $K$-file core with $N_1, N_2, \ldots, N_K$ systematic nodes for files $f_1, f_2, \ldots, f_K$, respectively, and $C$ coded nodes, with  demand $\lambda_1, \lambda_2, \ldots, \lambda_{K - 1}$  for files $f_1, f_2, \ldots, f_{K - 1}$. As in Theorem \ref{3files}, we again assume $\lambda_1 + \ldots + \lambda_{K-1} \leq \mu N_1 + \ldots + \mu N_{K - 1} + \frac{C}{K}\mu$. 
Our goal is to identify the maximal file $f_K$ request rate that can be served. 

We can first serve file $f_1, f_2, \ldots, f_{K - 1}$ demand\rmv{requests}   using their respective systematic nodes. This process is analogous to Step 1 in Theorem \ref{3files}. 
Note, in this algorithm,  the same demand\rmv{requests} is sent to every file $f_i$ systematic node, and also to every coded node, so we can let $\mu_i$ and $\mu_C$ represent the updated service rate of systematic file $f_i$ nodes and coded nodes, respectively. 

We can then serve any remaining total demand $\lambda = \lambda_1 + \ldots + \lambda_{K - 1}$ 
 using $K$-tuples of coded and systematic nodes. This is analogous to Step 2 in Theorem \ref{3files}. 
 Let $K' := \sum_{i= 1}^{K - 1} \sgn(N_i)$ denote the number of files (excluding file $f_K$) for which the system contains systematic nodes.
 There are $\left( \prod_{\{i = 1 \; | \; N_i > 0 \}}^{K-1} N_i \right) \binom{C}{K - K'}$ repair groups with $K'$  systematic nodes and $K - K'$ coded nodes. Letting $m$ be the index minimizing $N_i\mu_i$ for positive $N_i\mu_i$, 
 we can serve demand $\mu_m N_m$ 
 by sending $\frac{\mu_m}{\left( \prod_{\{i = 1 \; | \; N_i > 0 \text{ and } i \neq m \}}^{K-1} N_i \right)\binom{C}{K - K'}}$ demand\rmv{requests}  to each of these repair groups. 
 This exhausts file $f_m$ systematic nodes, while  file  $f_j$ systematic nodes ($j\neq m$, $N_j > 0$, $1 \leq j \leq K - 1$) have reduced service rate
\begin{multline}
\label{systematicupdate}
\mu_j - \frac{\mu_m\left( \prod_{\{i = 1 \; | \; N_i > 0 \text{ and } i \neq j \}}^{K-1} N_i \right) \binom{C}{K - K'}}{\left( \prod_{\{i = 1 \; | \; N_i > 0 \text{ and } i \neq m \}}^{K-1} N_i \right)\binom{C}{K - K'} } ,
\end{multline}
which is $\mu_j - \frac{\mu_m N_m}{N_j}.$
The remaining coded nodes have reduced service rate
\begin{multline}
\label{codedupdate}
\mu_C - \frac{\mu_m\left( \prod_{\{i = 1 \; | \; N_i > 0 \}}^{K-1} N_i \right) \binom{C - 1}{K - K'- 1}}{\left( \prod_{\{i = 1 \; | \; N_i > 0 \text{ and } i \neq m \}}^{K-1} N_i \right)\binom{C}{K - K'} } 
,
\end{multline}
which is $\mu_j - \frac{\mu_m (K - K')}{C}.$
We can continue in this way until the systematic node service rate is met for all but file $f_K$.
Then, we can use repair groups that consist entirely of coded nodes, applying Theorem \ref{allcodedK}. Once all  demand for files $f_1,\dots, f_{K - 1}$ has been satisfied, we can follow a similar process to utilize any remaining system resources to serve  demand for file $f_K$. Note, once  the  coded nodes have been exhausted, or if there are too few  coded nodes to form a $K$-tuple,  no demand may be satisfied using only coded nodes. We may then serve demand for  file $f_K$ using systematic file $f_K$ nodes.

\begin{algorithm}
\caption{Maximize $\lambda_K$}
\begin{algorithmic} 
\STATE \textbf{INPUT:} $\lambda_1, \lambda_2, \ldots, \lambda_{K-1}$, $N_1, N_2, \ldots, N_{K}, C, \mu$ 

\STATE \textbf{OUTPUT:} $\lambda_K$ 

\vspace{.05in}
\STATE $\lambda_K \leftarrow 0$ 

\STATE $\mu_C,\mu_i \leftarrow \mu$ for $i$ from $1$ to $K$  


\vspace{.05in}

\textbf{Step 1:} 

\vspace{.05in}

\FOR{$i$ from $1$ to $K - 1$} 
\IF{$\lambda_i \leq \mu N_i$} 

\STATE $\lambda_i \leftarrow 0$ 

\STATE $N_i \leftarrow N_i  - \left \lfloor \frac{\lambda_i}{N_i} \right \rfloor$ 

\STATE $\mu_i \leftarrow \mu  - \frac{\lambda_i - \left \lfloor \frac{\lambda_i}{\mu} \right \rfloor \mu }{N_i} $ 

\ELSE

\STATE $\lambda_i \leftarrow \lambda_i - \mu N_i$ 

\STATE $N_i, \mu_i \leftarrow 0$ 


\ENDIF

\ENDFOR

\vspace{.05in}

\textbf{Step 2:} 

\vspace{.05in}
\STATE $\lambda \leftarrow  \sum_{i = 1}^{K - 1} \lambda_i$ 
\STATE $K' \leftarrow \sum_{i = 1}^{K - 1} \sgn(N_i)$
\WHILE{$C > 0$ \textbf{and} $C \geq K - K'$} 

\IF{$K' > 0$} 

\STATE $m \leftarrow$ the index $i$ minimizing $N_i\mu_i, N_i\mu_i>0 $

\STATE $l \leftarrow \min(\mu_m N_m, \mu_C C)$ 

\IF{$\lambda > 0$} 

\IF{$\lambda \geq l$} 

\STATE $\lambda \leftarrow \lambda - l $ 

\ELSE

\STATE  $\lambda_K \leftarrow \lambda_K + (l - \lambda)$ 

\STATE  $\lambda \leftarrow 0$ 

\ENDIF

\ELSE

\STATE $\lambda_K \leftarrow \lambda_K + l$ 

\ENDIF

\IF{$l = \lambda_m N_m$} 

\STATE $\mu_C \leftarrow$ apply Equation \ref{codedupdate} 

\STATE $N_m,\mu_m \leftarrow 0$ 


\STATE $K' \leftarrow K' - 1$ 

\ELSE

\STATE $\mu_C, C \leftarrow 0$ 


\ENDIF



\STATE $\mu_j \leftarrow $ apply Equation \ref{systematicupdate} if $N_j > 0$ for $1 \leq j \leq K - 1$



\ELSE 

\IF{$\lambda > 0$}



\STATE $\mu_C \leftarrow \mu_C - \frac{\lambda}{\binom{C}{K}} \binom{C - 1}{K  - 1}$
\STATE $\lambda \leftarrow 0$ 
\ENDIF
\STATE $\lambda_K \leftarrow \lambda_K + \frac{C}{K} \mu_C$ 

\STATE $C \leftarrow 0$


\ENDIF

\ENDWHILE

\STATE $\lambda_K \leftarrow \lambda_K + \mu_K N_K$

\end{algorithmic}
\end{algorithm}





\section*{Acknowledgment}

The initial stages of this work were performed at ICERM (Institute for Computational and Experimental Research in Mathematics) in Providence, RI. We are indebted to the organizers of the ICERM 2017 Women in Data Science and Mathematics Research Collaboration Workshop.

\bibliographystyle{IEEEtran}
\bibliography{refs} 





\end{document}